\documentclass[12pt, english, a4paper,reqno]{amsart}

\usepackage{graphicx}
\usepackage{tikz}
\usepackage{capt-of}
\usepackage[applemac]{inputenc}
\usepackage{hyperref}

\usepackage{amsmath,amssymb}
\usepackage{bm}
\usepackage{euscript,color}
\usepackage{babel}
\usepackage{amsthm}
\usepackage{amsfonts}
\usepackage{latexsym}
\usepackage{fancyhdr}
\usepackage{enumerate}

\usepackage[numbers]{natbib}

\usepackage[left=2.5cm, top=3cm, bottom=3cm, right=2.5cm]{geometry}
\newcommand{\R}{\mathbb{R}}
\newcommand{\Q}{\mathbb{Q}}

\newcommand{\ve}{\underline e}
\newcommand{\vf}{\underline f}
\newcommand{\vbeta}{{\underline\beta}}
\newcommand{\vsigma}{\underline\sigma}
\newcommand{\eps}{\epsilon}
\newcommand{\si}{\mathfrak{s}}

\newcommand{\comment}[1]{}

\makeatletter
\renewcommand{\section}{\@startsection%
{section}
{1}
{0mm}
{1.5\bigskipamount}
{0.5\bigskipamount}
{\centering\normalsize\sc}}

\renewcommand{\paragraph}{\@startsection%
{paragraph}
{4}
{0mm}
{\bigskipamount}
{-1.25ex}
{\normalsize\sl}}

\def\provedboxcontents#1{$\square$}

\makeatother

\newtheoremstyle{thm}{}{}{\slshape}{}{\scshape}{.}{0.5em}{}

\newtheoremstyle{def}{}{}{}{}{\scshape}{.}{0.5em}{}

\newtheoremstyle{rmk}{}{}{}{}{\scshape}{.}{0.5em}{}

\newtheoremstyle{claim}{}{}{}{}{\slshape}{.}{0.5em}{}

\newtheorem*{thmGW}{Theorem GW} 

\newtheorem{newstatement}{newstatement}
\newtheorem{lemma}[newstatement]{Lemma}
\newtheorem{theorem}[newstatement]{Theorem}
\newtheorem{corollary}[newstatement]{Corollary}
\newtheorem{proposition}[newstatement]{Proposition}

\newtheorem*{conjecture*}{Conjecture}
\newtheorem*{prop*}{Proposition}

\theoremstyle{def}

\theoremstyle{rmk}
\newtheorem{remark}[newstatement]{Remark}

\theoremstyle{claim}

\expandafter\let\expandafter\oldproof\csname\string\proof\endcsname
\let\oldendproof\endproof
\renewenvironment{proof}[1][\proofname]{%
  \oldproof[\slshape #1]%
}{\oldendproof}

\let\geq\geqslant
\let\leq\leqslant

\let\phi\varphi
\let\epsilon\varepsilon

\renewcommand{\emph}[1]{{\slshape #1}}

\title[Risk aversion and Uniqueness]{Risk aversion and uniqueness of equilibrium: a polynomial approach}

\author{Andrea Loi}
\address{Andrea Loi, Dipartimento di Matematica e Informatica \\
         Universit\`a di Cagliari, Italy.}
         \email{loi@unica.it}

\author{Stefano Matta}
\address{Stefano Matta, Dipartimento di Scienze economiche e Aziendali \\
         Universit\`a di Cagliari, Italy.}
         \email{smatta@unica.it}

\date{}

\thanks{The first author was  supported  by INdAM. GNSAGA - Gruppo Nazionale per le Strutture Algebriche, Geometriche e le loro Applicazioni. Both authors were supported 
by STAGE - Funded by Fondazione di Sardegna.}

\begin{document}
\begin{abstract}
We study the connection between risk aversion, number of consumers and uniqueness of equilibrium.
We consider an economy with two goods and $c$ impatience types, where
each type has additive separable preferences with HARA Bernoulli utility function,
$u_H(x):=\frac{\gamma}{1-\gamma}\left(b+\frac{a}{\gamma}x\right)^{1-\gamma}$.
We show that if $\gamma\in \left(1, \frac{c}{c-1}\right]$, the equilibrium is unique.
Moreover, the methods used, involving Newton's symmetric polynomials and Descartes' rule of signs,
enable us to offer new sufficient conditions for uniqueness in a closed-form expression highlighting 
the role played by endowments, patience and specific HARA parameters.
Finally, new necessary and sufficient conditions in ensuring uniqueness 
are derived for the particular case of CRRA Bernoulli utility functions with $\gamma =3$.

\end{abstract}
\maketitle

\vspace{0.3in}


\vspace{0.3in}


\vspace{0.3in}

\section*{Acknowledgement}

\thispagestyle{empty}

We would like to acknowledge Alexis Akira Toda for useful comments and critical remarks.

\section{Introduction and statements of the main results}

Uniqueness plays a crucial role in comparative statics and stability.
Yet, as highlighted by \cite{ke98}, it is rarely possibile to provide easy analytical conditions
that guarantee uniqueness in applied models.
In a recent paper, Geanakoplos and Walsh \cite{gewa}  have presented new sufficient conditions ensuring uniqueness and stability of equilibrium
in an economy with two goods, where $c$ agents, $c\geq 2$,  ordered according to a parameter $\vbeta=(\beta_1,\ldots, \beta_c$), $\beta_1<\cdots <\beta_c$, representing patience, have identical endowments and the same Bernoulli utility function displaying non-increasing absolute risk aversion (DARA).
The role played in that paper by the DARA assumption is twofold: it ensures that income effects can be ordered across types and
determines a positive covariance between consumption and income effect, hence bounding the market income effect.

 For our purposes, their main results  (see Proposition 2 and Proposition 5 in  \cite{gewa} for details)
 can be summarized  as follows.

\begin{thmGW}\label{thmGW}
Let $u$ be a Bernoulli utility function and let impatience type $i$'s preferences be represented by
$$
u_i(x,y)=u(x)+\beta_iu(y),\ i=1, \dots ,c.
$$
Denote by $(e_i, f_i)$ consumer $i$'s endowments of goods $x$ and $y$, respectively.

The price is unique if:
\begin{enumerate}
\item
 $u$ satisfies  DARA and agents have  identical endowments, i.e. $(e_i, f_i)=(e_j, f_j)$ for all $i$ and $j$;
 \item
$u$ satisfies  CRRA and the following restrictions on patience and endowments hold: $e_i\leq e_j$ and $f_i\geq f_j$, for any  $i<j$.
 \end{enumerate}
\end{thmGW}

They point out that the assumption of identical endowments, although used in several papers, is rather restrictive
and highlight that, under DARA preferences, there is no evident condition which ensures uniqueness
if the assumption of identical endowments is dropped.
They conjecture that there shouldn't be \lq\lq too much heterogeneity'' across agents to ensure uniqueness,
arguing that heterogeneity should involve a condition on the patience parameter $\vbeta$, endowments
and the particular Bernoulli utility function used.

In this paper we are mainly interested in exploring sufficient conditions on the parameter $\gamma$ 
that guarantees uniqueness of equilibrium without imposing any restriction on endowments, unlike Theorem GW.
Moreover, the methods used enable us to address the above issue of heterogeneity raised by \cite{gewa} 
and to obtain necessary and sufficient conditions for CRRA preferences.

More precisely, we consider an economy with two goods and $c$ impatience types,
where  type  $i$ has preferences represented by the utility function
\begin{equation}\label{addsep}
u_i(x,y)=u_H(x)+\beta_iu_H(y),
\end{equation}
where   $u_H$ is HARA\footnote{We will not consider the well-known case $\gamma=1$, where the Bernoulli utility is logarithmic.}, i.e.
\begin{equation}\label{uH}
u_H(x):=\frac{\gamma}{1-\gamma}\left(b+\frac{a}{\gamma}x\right)^{1-\gamma},\ \gamma>0, \gamma \neq 1, a>0,  b\geq 0.
\end{equation}

Notice that HARA is an important subclass of DARA preferences extensively used in the literature, which
also encompasses the CRRA case by setting $b=0$. 

 Our main result is the following theorem which establishes a connection between uniqueness,
$\gamma$ and the number $c$ of impatience types in the economy.

\begin{theorem}\label{mainteor}
Let $u_H$ be the  Bernoulli utility function \eqref{uH} and let utility of impatience type  $i$  be  given by \eqref{addsep}.
If 
\begin{equation}\label{assgamma}
\gamma\in \left(1, \frac{c}{c-1}\right],
\end{equation}
then the economy has a unique equilibrium.
\end{theorem}

This result is also made interesting by the fact that
\lq\lq utility function concavity parameters in the range of $1$ to $2$ are widely considered plausible in the literature",
as Kaplow observes (see \cite[Chapter 3]{kap} and references therein).

For CRRA preferences this links the coefficient of relative risk aversion $\gamma$ to the number of consumers
and also answers the issue raised by \cite{towa}, i.e. whether a value of $\gamma$ in the interval
$(1,2]$ is compatible with multiple equilibria for CRRA preferences with $2$ consumers.
Moreover, it provides a generalization for any number of consumers and
arbitrary endowments allocations.

In fact, it is known (see \cite{msc91} and references therein; see also \cite{helo} and \cite{msc95}) that for $C^2$ separable preferences $\sum_l u_l(x)$,
where each $u_l$ is a monotonic and concave, relative risk aversion less or equal to $1$ implies uniqueness. In particular, 
under CRRA preferences, i.e. when $b=0$, 
if $\gamma$ belongs to $(0,1)$, we have uniqueness. 
Toda and Walsh \cite{towa} show for CRRA preferences the existence of multiple equilibria
in an economy with two goods and two consumers when $\gamma>2$. Whether multiplicity is possibile or not
for $1<\gamma\leq 2$ was left an open question. An immediate consequence of Theorem \ref{mainteor}, see Corollary
\ref{todaq}, rules out this possibility. Moreover, it highlights the interplay between relative risk aversion and the number of consumers.
An equivalent result also holds for HARA preferences (see Corollary \ref{gewaq}).

As far as the heterogeneity issue raised by \cite{gewa} is concerned, 
our contribution is the following theorem.

\begin{theorem}\label{mainteor2}
In the HARA case with two impatience types  and for $\gamma$
sufficiently close to $3$, we have at most three equilibria. Moreover, if we  assume that
\begin{equation}\label{prop5}
\beta_1<\beta_2, \  e_1\leq e_2, \ f_1\geq f_2,
\end{equation}
and
\begin{equation}\label{bmag}
b\geq\frac{a}{3}\left(\frac{\beta_2}{\beta_1}\right)^{\frac{2}{3}}\left(e_2+f_1\right)
\end{equation}
are satisfied, then the  equilibrium  is unique.
\end{theorem}

Condition \eqref{bmag} confirms what has been claimed by \cite{gewa}.
It represents a closed-form expression that embodies what \cite{gewa} argue about for arbitrary DARA preferences, highlighting the role played by endowments, patience, and specific Bernoulli utility parameters in ensuring uniqueness. In particular, the parameters involved $a,b,\gamma$ affect risk tolerance, i.e. the inverse of absolute risk aversion. Following \cite{gewa}'s insight, this closed-form expression ensures the effect of positive covariance across types.
Moreover, it allows for more heterogeneity in allocations than might be expected.

Since we believe that the approach used is as interesting as the results, we provide an intuition here, leaving the details in Section \ref{harasec}.	
As usual, from the maximization problem, one obtains the excess demand function $Z$, expressed as an implicit function of the price raised to a positive function
depending on the parameter $\gamma$. The strategy is to turn $Z$, by algebraic manipulations, into a polynomial $P$. Since $\gamma$ can be a real number, we use the density of $\Q$
in $\R$ to approximate $\gamma$. This strategy \lq\lq generically'' works for regular values for a transversality argument.
Finally, we write $P$ in terms of  Newton's symmetric polynomials and apply Descartes' rule of signs, which states that the number of positive roots
of a real polynomial, arranged in ascending or
descending powers, cannot exceed the number of sign
variations in consecutive coefficients (see, e.g., \cite{ajs}).

It is remarkable how this simple method can be very useful and powerful. In fact, it allows to study 
the number of equilibria without ad hoc restrictions on the set of parameters and, moreover,
it can provide sufficient and necessary conditions for uniqueness as we show
in Theorem \ref{mainteor3}, which establishes a connection between the number of equilibria, 
their type (regular and critical), endowments and utility weights, thus complementing 
\cite{towa}'s analysis of a CRRA economy with relative risk aversion $\gamma=3$.

The literature on uniqueness is vast. For  a survey,
the reader is primarily referred to \cite{ke98, msc91}, and references therein.
As we have already pointed out, a feature of Theorem \ref{mainteor} is that we do not impose restrictions on endowments.
From this point of view, it is related to the strand of the literature 
which provides sufficient 
conditions that guarantee uniqueness globally, i.e. for every possible allocation of resources among consumers. 
Two recent papers that belong to this line of research are \cite{lm18,lm21}, 
where uniqueness is globally characterized through geometric properties of the equilibrium manifold.
For a different approach, which provides sufficient conditions on offer curves 
in a two-commodity, two-agent exchange economy, see \cite{gime}.
We finally refer the reader to \cite{krs}, and references therein,
for a survey on how to solve economic equilibria described as 
solutions to systems of polynomial equations.

We believe that the uniqueness issue, and the methods used in this paper, can deserve further attention and research. 
In particular, we think about extending Theorem \ref{mainteor2} to an arbitrary number of consumers  without imposing restrictions on $\gamma$.
Moreover, it could be interesting to apply these methods to different Bernoulli utility functions, in order to achieve new sufficient conditions on uniqueness.

This paper is organized as follows. In Section \ref{harasec} we describe the economic setting for HARA preferences, the methods used and we prove Theorem 
\ref{mainteor}. Section \ref{crrasec} applies our approach to the particular case of $c=2$ consumer types (Corollary \ref{gewaq} and \ref{todaq}) 
and provides new sufficient (Theorem \ref{mainteor2}) and necessary and sufficient conditions for uniqueness (Theorem \ref{mainteor3}),
for HARA and CRRA preferences, respectively.

\section{HARA Preferences}\label{harasec}

We consider a pure exchange economy with two goods and an arbitrary (finite) 
number ($c$) of impatience types, where  type $i$ has preferences represented by \eqref{addsep}.

Consumer $i$'s endowments is denoted by $(e_i,f_i)$ . We have $\sum_{i=1}^ce_i=r_x(c)$ and $\sum_{i=1}^cf_i=r_y(c)$,
where $(r_x(c),r_y(c))$ is the total resources vector.

Under the budget constraint $pe_i+f_i\leq px+y$, 
consumer $i$'s maximization problem, for $i=1,\ldots, c$, leads to the aggregate excess demand function for good $x$:

\begin{equation}\label{z1}
Z\left(\ve, \vf,p,\epsilon,\vsigma,a,  b, c\right):=\sum_{i=1}^c\frac{b-bp^\epsilon\sigma_i+a\epsilon\left(pe_i+f_i\right)}{a\epsilon\left(p+\sigma_i p^\epsilon \right)}-r_x(c),
\end{equation}
where we set 
$$\epsilon:=\frac{1}{\gamma}, \ \sigma_i:=\beta_i^{\epsilon}, \ i=1,\dots,c$$
and we denote
$$\ve=\left(e_1,\dots, e_c\right),\vf=\left(f_1,\dots, f_c\right), \vsigma=\left(\sigma_1,\dots, \sigma_c\right).$$

In order to make algebraic manipulations of expression \eqref{z1} easier (see Proposition \ref{mainprop} below), we will exploit 
the presence of symmetric polynomials within \eqref{z1} via Newton's identities. 
The interested reader may refer to \cite{edwa}, although this paper is self-contained.

For each integer $t$ such that $1\leq t\leq c$ and  $i=1,\dots,c$ set
\begin{equation}\label{sigma}
\si\left(t, c\right):=\sum_{i_1<\cdots <i_t}^{c} \sigma_{i_1\dots i_t},  \  i_l=1, \dots, c, \ l=1, \dots ,t,
\end{equation}
where 
$$\sigma_{i_1\dots i_t}=\sigma_{i_1}\cdots \sigma_{i_t}.$$
Moreover, set
\begin{equation}\label{zerobis}
\si(c, c, j)=0
\end{equation}
and for $1\leq t <c$ 
\begin{equation}\label{sigmarid}
\si (t, c, i):=\sum_{i_1(i)<\cdots <i_t(i)}^c \sigma_{i_1(i)\dots i_t(i)},  \  i_l(i)=1, \dots \hat i, \dots c, \ l=1, \dots ,t,
\end{equation}
where $\hat i$ means that the index $i$ is omitted.

 Further,   set 
\begin{equation}\label{zero}
\si (0, c)=\si (0, c, i)=1 
\end{equation}
and
\begin{equation}\label{zerogood}
\si (c, c-1)=\si(c, c-1, i)=0 .
\end{equation}
Then it is not hard to see that   the following equalities hold true:
\begin{equation}\label{sigma1}
\si(t, c)=\si (t, c-1)+\sigma_c\si(t-1, c-1),
\end{equation}
\begin{equation}\label{sigma2}
\si(t, c, i)=\si(t, c-1, i)+\sigma_c\si(t-1, c-1, i),
\end{equation}
\begin{equation}\label{sigma3}
\si (t, c, c)=\si(t, c-1),
\end{equation}
for all $c\geq 2$,  $1\leq t\leq c$ and $i=1,\dots, c$.

\begin{proposition}\label{mainprop}
The zeros set  of the aggregate excess demand function (\ref{z1}) equals 
that  of the following function:
\begin{equation}\label{superz}
\begin{split}
z(\ve, \vf,p,\epsilon,\vsigma,a,  b, c):=-\sigma_{1\dots c}\left(r_x(c)+\frac{cb}{a\epsilon}\right)p^{c(\eps - 1)+1}+\\
 -\sum_{t=1}^{c-1}\left[\xi(\ve, \vsigma,  c,  t)+u(\vsigma, b, c,  t)\right]p^{t(\epsilon-1)+1}\\
+\sum_{t=1}^{c-1}v(\vf, \vsigma, b, c,  t)p^{t(\epsilon-1)}+r_y(c)+\frac{cb}{a\epsilon}
\end{split}
\end{equation}

where
\begin{equation}\label{xi}
\xi(\ve, \vsigma,  c, t):=\left[r_x(c)\si (t, c) -\sum_{i=1}^{c} e_i\si (t, c, i)\right],
\end{equation}
\begin{equation}\label{u}
u(\vsigma, a, b, c,  t):=\frac{b}{a\epsilon}\sum_{i=1}^{c}\sigma_i\si (t-1, c, i),
\end{equation}
\begin{equation}\label{v}
v(\vf, \vsigma, a,  b, c, t):= \sum_{i=1}^{c}\left(f_i+\frac{b}{a\epsilon}\right)\si(t, c,i).
\end{equation}
\end{proposition}
\begin{proof}
At an equilibrium price $p$ the function \eqref{z1} 

$$
-pr_x(c)+\sum_{i=1}^c\frac{pe_i+f_i+\frac{b}{a\epsilon}-\frac{b}{a\epsilon}\sigma_i p^\eps}{1+\sigma_ip^{\epsilon-1}} ,$$

or equivalently 

$$- p r_x(c) \prod_{i=1}^c\left(1+\sigma_ip^{\epsilon-1}\right)+\sum_{i=1}^{c}\left[\left(pe_i+f_i+\frac{b}{a\epsilon}-\frac{b}{a\epsilon}\sigma_ip^\eps\right)\prod_{\hat i=1}^c\left(1+\sigma_{\hat i}p^{\epsilon-1}\right)\right],$$
where  $\hat i$ means that the index $i$ is omitted, vanishes.

By Equations (\ref{sigma}) and (\ref{sigmarid}), we can write the products above as follows

$$\prod_{i=1}^c(1+\sigma_ip^{\epsilon-1})=1+\sum_{t=1}^{c-1}\si(t,c)p^{t(\epsilon-1)}+\sigma_{1\ldots c}p^{c(\epsilon-1)}$$

$$\prod_{\hat i=1}^c(1+\sigma_{\hat i}p^{\epsilon-1})=1+\sum_{t=1}^{c-1}\si(t,c,i)p^{t(\epsilon-1)}$$

and rewrite the expression accordingly:

$$-pr_x(c)\left[1+\sum_{t=1}^{c-1}\si(t,c)p^{t(\epsilon-1)}+\sigma_{1\ldots c}p^{c(\epsilon-1)}\right]+\sum_{i=1}^{c}\left(pe_i+f_i+\frac{b}{a\epsilon}-\frac{b}{a\epsilon}\sigma_ip^\epsilon\right)\left[1+\sum_{t=1}^{c-1}\si(t,c,i)p^{t(\epsilon-1)}\right].$$

By expanding and rearranging, we immediately  get

\begin{align*}
 &-r_x(c)\sigma_{1\ldots c}p^{c(\epsilon-1)+1}-\sum_{t=1}^{c-1}\left[r_x(c)\si(t,c)-\sum_{i=1}^{c}e_i\si(t,c,i)\right]p^{t(\epsilon-1)+1}+\sum_{t=1}^{c-1}\sum_{i=1}^{c}(f_i+\frac{b}{a\epsilon})\si(t,c,i)p^{t(\epsilon-1)}\\
&-\frac{b}{a\epsilon}(\sum_{i=1}^c\sigma_i)p^\epsilon
-\frac{b}{a\epsilon}\sum_{t=1}^{c-1}\sum_{i=1}^{c}\sigma_i\si(t,c,i)p^{t(\epsilon-1)+\epsilon}+r_y(c)+\frac{cb}{a\epsilon}.
\end{align*}

Notice now that by the change of index $u:=t+1$, one gets

\begin{align*}
&-\frac{b}{a\epsilon}\sum_{t=1}^{c-1}\sum_{i=1}^{c}\sigma_i\si(t,c,i)p^{t(\epsilon-1)+\epsilon}=-\frac{b}{a\epsilon}\sum_{t=1}^{c-1}\sum_{i=1}^{c}\sigma_i\si(t,c,i)p^{(t+1)(\epsilon-1)+1}\\
&=-\frac{b}{a\epsilon}\sum_{u=2}^{c}\sum_{i=1}^{c}\sigma_i\si(u-1,c,i)p^{u(\epsilon-1)+1}=
\frac{b}{a\epsilon}\sum_{i=1}^{c}\sigma_i\si(0,c,i)p^{(\epsilon-1)+1}\\
&-\frac{b}{a\epsilon}\sum_{t=1}^{c-1}\sum_{i=1}^{c}\sigma_i\si(t-1,c,i)p^{t(\epsilon-1)+1}-\frac{b}{a\epsilon}\sum_{i=1}^{c}\sigma_i\si(c-1,c,i)p^{c(\epsilon-1)+1}\\
&=\frac{b}{a\epsilon}(\sum_{i=1}^c\sigma_i)p^\epsilon-\frac{b}{a\epsilon}\sum_{t=1}^{c-1}\sum_{i=1}^{c}\sigma_i\si(t-1,c,i)p^{t(\epsilon-1)+1}-\frac{cb}{a\epsilon}\sigma_{1\ldots c}p^{c(\epsilon-1)+1},
\end{align*}

where in the last equality we use \eqref{zero} and $\sum_{i=1}^{c}\sigma_i\si(c-1,c,i)=c\sigma_{1\ldots c}$, where $\sigma_{1\ldots c}=\sigma_{1}\cdots \sigma_{c}$.

By inserting this last equality into the previous expression one gets
\begin{align*}
&-\sigma_{1\dots c}\left(r_x(c)+\frac{cb}{a\epsilon}\right)p^{c(\eps - 1)+1}-\sum_{t=1}^{c-1}\left[r_x(c)\si(t,c)-\sum_{i=1}^{c}e_i\si(t,c,i)+\frac{b}{a\epsilon}\sum_{i=1}^{c}\sigma_i\si(t-1,c,i)\right]p^{t(\epsilon-1)+1}\\
&+\sum_{t=1}^{c-1}\sum_{i=1}^{c}\left(f_i+\frac{b}{a\epsilon}\right)\si(t,c,i)p^{t(\epsilon-1)}+r_y(c)+\frac{cb}{a\epsilon},
\end{align*}
and the proposition follows.
\end{proof}

The following two lemmata are needed to prove Theorem \ref{mainteor}.

\begin{lemma}\label{lemmafond}
Set
\begin{equation}\label{infond}
F(t, c):=r_x(c)\si (t, c)-\sum_{i=1}^ce_i\si (t, c, i).
\end{equation}
Then $F(t, c)>0$
 for each integer $t$ such that $1\leq t\leq c-1$. 
\end{lemma}
\begin{proof}
We work by induction on $c\geq 2$ for all  $t$ such that $1\leq t\leq c-1$.
The base on the induction is immediate:
\begin{align*}
F(1, 2)&=(e_1+e_2)\sigma (1, 2)-e_1\sigma (1, 2, 1)-e_2\sigma (1, 2, 2)\\
&=(e_1+e_2)(\sigma_1+\sigma_2)-e_1\sigma_2-e_2\sigma_1\\
&=e_1\sigma_1+e_2\sigma_2>0\\
F( 2, 2)&=(e_1+e_2)\sigma (2, 2)-e_1\sigma (2, 2, 1)-e_2\sigma (2, 2, 2)=(e_1+e_2)\sigma_1\sigma_2 >0.\\
\end{align*}
Assume now, by the induction assumption, that 
$$
F(t, c-1)>0
$$
for each integer  $1\leq t\leq c-2$.
By \eqref{sigma1}, \eqref{sigma2} and  \eqref{sigma3} equation \eqref{infond} reads  as
$$F(t, c)=r_x(c-1)\si(t, c-1)+\sigma_cr_x(c)\si (t-1, c-1)-\sum_{j=1}^{c-1}e_j\left[\si(t, c-1, j)+\sigma_c \si (t-1, c-1, j)\right],$$
which can be rewritten as
$$F(t, c)=F(t, c-1)+\sigma_cF(t-1, c-1)+\sigma_ce_c\si(t-1, c-1),$$
which is strictly positive by the induction assumption.
\end{proof}

The next lemma can only be applied to regular equilibra, being robust to sufficiently small perturbations. 
Let us denote by $N(\ve , \vf ,\epsilon,\vsigma, a, b, c)$  the (finite) cardinality of the set of regular equilibria 
of the aggregate excess demand \eqref{z1}, or equivalently \eqref{superz}, in the Hara case.

\begin{lemma}\label{lemmac} 
For every $\epsilon_0\in (0,1)$, there exist two natural numbers $m,\, n$, $0<m<n$, 
where $\frac{m}{n}$ is sufficiently close to $\epsilon_0$, such that
$$N(\ve , \vf ,\epsilon_0,\vsigma, a, b, c)\leq N(\ve,\vf, \frac{m}{n},\vsigma,a,  b, c),$$
for all
$\ve,\vf,\vsigma, a,  b, c$.
\end{lemma}
\begin{proof}

For fixed $(\ve,\vf,\vsigma, a,  b, c)$, the aggregate excess demand function \eqref{superz} 
continuously depends on the price $p$ and the parameter $\epsilon$. 
Notice that if $p_0$ is a regular equilibrium of the function
of one variable
$z\left(\ve, \vf,p,\epsilon_0,\vsigma,a,  b, c\right)$,
then its graph transversally intersects
the $p$ axis in the  $p-z$ plane. 
The same property holds for a small perturbation of $\epsilon_0$.
Hence the result follows by the density of $\Q$ in $\R$.
 \end{proof}

This lemma points out that the number of regular equilibra cannot decrease
after the perturbation. The following figure provides an insight into the previous lemma. The black curve represents the aggregate  excess demand curve for a given $\eps$. 
The perturbation induced in the original curve by replacing  $\eps$ with $\frac{m}{n}$ is represented by the red curve. 
In this case the number of regular equilibria is such that
$$3=N(\ve , \vf ,\epsilon,\vsigma, a, b, c)\leq N(\ve,\vf, \frac{m}{n},\vsigma,a,  b, c)=5.$$
Observe that if the perturbation had  \lq\lq opposite direction",
the number of regular equilibria would remain unchanged.

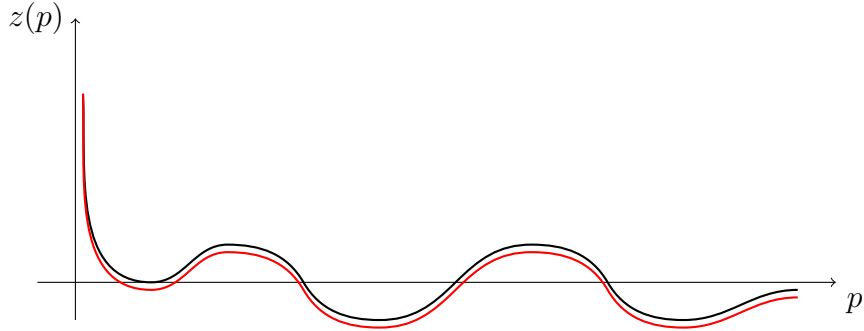
\begin{figure}[h]
\begin{tikzpicture}
\centering
\draw[thick,color=black,step=.5cm,dashed] (-0.5,-.5); 
\draw[->] (-0.5,0) -- (10,0)
node[below right] {$p$};
\draw[->] (0,-0.5) -- (0,3.5)
node[left] {$z(p)$};
\draw [thick, smooth] (0.1,2.5)  to[out=-88, in=180] (1,0)  
to[out=0, in=180] (2,0.5) to[out=0, in=120](3,0) 
to[out=-60, in=180](4,-0.5) to[out=0, in=225](5,0) 
to[out=45, in=180](6,0.5) to[out=0, in =120](7,0) 
to[out=-60, in =180](8,-0.5) to[out=0, in =180](8,-0.5) 
to[out=0, in =180](9.5,-0.1); 
\draw [thick, red, smooth] (0.1,2.5)  to[out=-88, in=180] (1,-0.1) 
to[out=0, in=180] (2,0.4) to[out=0, in=120](3,-0.1) 
to[out=-60, in=180](4,-0.6) to[out=0, in=225](5,-0.1) 
to[out=45, in=180](6,0.4) to[out=0, in =120](7,-0.1)
 to[out=-60, in =180](8,-0.6) to[out=0, in =180](8,-0.6) 
 to[out=0, in =180](9.5,-0.2); 
\end{tikzpicture}
\caption{Aggregate excess demand after replacing $\eps$ with $\frac{m}{n}$.}
\label{figz}
\end{figure}

\vspace{1cm}

We are now in position to prove Theorem \ref{mainteor}.

\begin{proof}[Proof of  Theorem \ref{mainteor}]
By inserting  $\epsilon=\frac{m}{n}$, rational number, as in Lemma \ref{lemmac} into \eqref{superz} and denoting $q=p^{\frac{1}{n}}$,
we deduce that 
\begin{equation}\label{superzz}
\begin{split}
z(\ve, \vf,q,\frac{m}{n},\vsigma, a,  b, c)&:=-\sigma_{1\dots c}(r_x(c)+\frac{cb}{a\epsilon})q^{c(m-n)+n}\\ 
&-\sum_{t=1}^{c-1}\left[\xi(\ve, \vsigma, c,   t)+u(\vsigma, a, b, c,  t)\right]q^{t(m-n)+n}\\
&+\sum_{t=1}^{c-1}v(\vf, \vsigma, a, b, c,  t)q^{t(m-n)}+r_y+\frac{cb}{a\epsilon}
\end{split}
\end{equation}

We claim that for $\frac{m}{n} \in \left[\frac{c-1}{c}, 1 \right)$ and for $\ve, \vf,\vsigma,a,  b, c$ arbitrarly chosen,  there exists  a unique  $q_0>0$ which is a zero of the function
\begin{equation}\label{zmod}
z(\ve, \vf,q,\frac{m}{n},\vsigma,a,  b, c), 
\end{equation}
i.e. $z(\ve, \vf,q_0,\frac{m}{n},\vsigma, a, b, c)=0$.
Then,  by applying Lemma \ref{lemmac},
we would deduce that $p_0=q_0^n$ is the only positive solution of  \eqref{superz} for $\eps \in [\frac{c-1}{c}, 1 )$ , i.e. we have unicity of price equilibrium in 
the HARA economy under the assumption
\eqref{assgamma}, thus arriving at a conclusion of the proof of the theorem.

In order to prove the claim,  fix $\ve, \vf,\vsigma, a, b, c$ and set
$$\mu_c:=\sigma_{1\dots c}(r_x(c)+\frac{cb}{a\epsilon}), \  \mu_t:=\xi_t + u_t,\ v_t, \ v_0=r_y+\frac{cb}{a\epsilon},$$
where for  $t=1, \dots c-1$ we define
$$\xi_t:=\xi(\ve, \vsigma,  c, t),\  u_t:=u(\vsigma, a, b, c,  t),\  v_t:=v(\vf, \vsigma, a, b, c, t).$$
By multiplying equation \eqref{zmod} by the monomial $q^{(c-1)(n-m)}$ ($n>m$) and using \eqref{superzz},
one sees that there exists a positive zero of 
\eqref{zmod} if and only if the following polynomial in the variable $q$ has a unique positive root:

\begin{equation}\label{poli}
\begin{split}
P(q)&:=v_{c-1}+v_{c-2} q^{n-m}+\dots +v_1q^{(c-2)(n-m)}+v_0q^{(c-1)(n-m)}\\
&-\mu_cq^m-\mu_{c-1}q^n-\dots -\mu_1q^{(c-2)(n-m)+n}
\end{split}
\end{equation}

Notice now that by the very definition of the symbols involved, $\mu_c$, $\nu_0$, $u_t$ and $v_t$ are strictly positive real numbers for all $t=1, \dots ,c-1$.
By  applying Lemma \ref{lemmafond},  we also deduce that $\xi_t$ is strictly positive and hence 
$\mu_t>0$ for all $t=1, \dots ,c-1$. Since the assumption $\frac{m}{n} \in [\frac{c-1}{c}, 1 )$ is equivalent to $(c-1)(n-m)\leq m$, 
it follows that the monomials appearing in $P(q)$
are written in  increasing order. Thus, by  Descartes' rule of sign, the polynomial $P(q)$ has a unique positive root and the theorem follows.
\end{proof}

\section{The case of two consumers}\label{crrasec}

In this  section we the study the number of equilibria in the  case of $c=2$ consumers.
Two immediate implications  of Theorem \ref{mainteor} are the following corollaries.
\begin{corollary}\label{gewaq}
In the HARA case with two goods, if  $\gamma$ belongs to
the interval $(1,2]$, then the equilibrium price is unique.
\end{corollary}
\begin{corollary}\label{todaq}
In the CRRA case with two goods, if the relative risk aversion $\gamma$ belongs to
the interval $(1,2]$, then the equilibrium price is unique.
\end{corollary}
This last corollary answers the question left open by Toda and Walsh \cite[Remark 1]{towa}
 of whether multiplicity is possibile or not
for $\gamma\leq 2$ (see Subsection \ref{CRRApart} below for a deeper analysis of Toda-Walsh example).

In the sequel, we are interested to analyze the number of equilibria in the HARA case when $\gamma =3$.

By the proof of Theorem \ref{mainteor}, we have reduced our investigation to the analysis of the zeros  of the polynomial 
\eqref{poli} when $c=2$, $m=3$, $n=1$ ($\eps=\frac{1}{3}$). Namely, by \eqref{xi}, \eqref{u} and  \eqref{v}, one has
\begin{equation}\label{polic=2}
P(q)=A(\ve, \vsigma, a, b)q^3+B(\vf, \vsigma, a, b)q^2+C(\ve, \vsigma, a, b)q+D(\vf, \vsigma, a, b),
\end{equation}
where 
\begin{equation}\label{ABCD}
\begin{alignedat}{3}
A(\ve, \vsigma, a, b):=&-\mu_1= -(e_1\sigma_1+e_2\sigma_2)-\frac{3b}{a}(\sigma_1+\sigma_2),\\
B(\vf, \vsigma, a, b):=&v_0= f_1+f_2+\frac{6b}{a},\\ 
C(\ve, \vsigma, a, b):=&-\mu_2= -(e_1+e_2)\sigma_1\sigma_2-\frac{6b}{a}\sigma_1\sigma_2, \\ 
D(\vf, \vsigma, a, b):=&v_1= f_1\sigma_2+f_2\sigma_1+\frac{3b}{a}(\sigma_1+\sigma_2).
\end{alignedat}
\end{equation}

In order to prove Theorem \ref{mainteor2}, we need  a simple  algebraic lemma.
\begin{lemma}\label{lemmapol}
Assume that the polynomial 
$P(x)=Ax^3+Bx^2+Cx+D$ has three sign changes and $ABCD\neq 0$.
If  
\begin{equation}\label{assABCD}
 AD-BC<0,
\end{equation}
then $P(x)$ has a unique positive root.
\end{lemma}
\begin{proof}
The discriminant of $P(x)$ is given by:
\begin{equation}\label{dis}
\Delta=B^2C^2-4AC^3-4B^3D-27A^2D^2+18ABCD.
\end{equation}
The assumption \eqref{assABCD} gives $A^2D^2>B^2C^2$ and $A^2D^2>ABCD$. 
Then $\Delta<0$. Hence the polynomial has a unique root which is positive by the Decartes' rule of signs (or by the mean value theorem).
\end{proof}
\begin{proof}[Proof of Theorem \ref{mainteor2}]
By Lemma \ref{lemmac}, one can assume that $\gamma=\frac{1}{\epsilon}=3$.
The boundedness of the number of equilibria is immediate, since the polynomial
\eqref{polic=2} has degree three.
In order to prove the second part of the theorem, 
by Lemma \ref{lemmapol} we need to verify that \eqref{prop5} and \eqref{bmag} imply 
$$A(\ve, \vsigma, a, b)D(\vf, \vsigma, a, b)-B(\vf, \vsigma, a, b)C(\ve, \vsigma, a, b)<0$$ 
where $A(\ve, \vsigma, a, b), B(\vf, \vsigma, a, b), C(\ve, \vsigma, a, b),  D(\vf, \vsigma, a, b)$ are given by \eqref{ABCD} above.
A long but  straightforward computation yields
\begin{equation}\label{dadim}
\begin{split}
A(\ve, \vsigma, a, b)D(\vf, \vsigma, a, b)-B(\vf, \vsigma, a, b)C(\ve, \vsigma, a, b)&=\\
&=(\sigma_2-\sigma_1)(e_1f_2\sigma_1-e_2f_1\sigma_2)\\&
+E(\ve,\vf, \vsigma, a, b),
\end{split}
\end{equation}

where 
$$E(\ve,\vf, \vsigma, a, b):=-\frac{9b^2}{a^2}\left(\sigma_1-\sigma_2\right)^2+\frac{3b}{a}\left[\left(e_1+e_2+f_1+f_2\right)\sigma_1\sigma_2-(e_1+f_2)\sigma_1^2 -(e_2+f_1)\sigma_2^2\right].$$

Notice that by \eqref{prop5}, the first summand of \eqref{dadim}, namely $(\sigma_2-\sigma_1)(e_1f_2\sigma_1-e_2f_1\sigma_2)$, is strictly less than zero:
$$(\sigma_2-\sigma_1)(e_1f_2\sigma_1-e_2f_1\sigma_2)\leq (\sigma_2-\sigma_1)f_1(e_1\sigma_1-e_2\sigma_2)< (\sigma_2-\sigma_1)f_1\sigma_2(e_1-e_2)\leq 0.$$

Moreover, $E(\ve,\vf, \vsigma, a, b)\leq 0$ if and only if 
$$b\geq\frac{a}{3} \frac{\left[(e_1+e_2+f_1+f_2)\sigma_1\sigma_2-(e_1+f_2)\sigma_1^2 -(e_2+f_1)\sigma_2^2\right]}{\left(\sigma_1-\sigma_2\right)^2}.$$

Again by \eqref{prop5} one can find an upper bound of  the right hand side of this inequality, namely:

\begin{align*}
\frac{a}{3}\frac{\left[(e_1+e_2+f_1+f_2)\sigma_1\sigma_2-(e_1+f_2)\sigma_1^2 -(e_2+f_1)\sigma_2^2\right]}{\left(\sigma_1-\sigma_2\right)^2}& \\
&<\frac{a}{3}\frac{\left[(e_1+e_2+f_1+f_2)\sigma_1\sigma_2\right]}{\left(\sigma_1-\sigma_2\right)^2}\\&<
\frac{a}{3}(\frac{\sigma_2}{\sigma_1})^{2}(e_2+f_1).
\end{align*}

Thus, if \eqref{prop5} and   $b\geq \frac{a}{3}(\frac{\sigma_2}{\sigma_1})^{2}(e_2+f_1)$ hold true, we get that \eqref{dadim} is strictly less than zero and by Lemma \ref{lemmapol} there is uniqueness  of equilibria.
Hence the  proof of the theorem follows,  since  $\sigma_i=\beta_i^{\frac{1}{3}}$, $i=1, 2$.
\end{proof}

\begin{remark}\label{remarkprop5}
Conditions \eqref{prop5}
are exactly those of \cite[Proposition 5]{gewa}, which holds for CRRA (homotethic) preferences.
In fact, assuming only conditions \eqref{prop5} and CRRA preferences, one can follow
the same argument of the proof of Theorem \ref{mainteor2} and achieve Proposition 5 in a different way.
Hence Theorem \ref{mainteor2} encompasses Proposition 5 as a special case, when $\gamma=3$
(see Remark \ref{remarkresearch}).\footnote{We thank again Alexis Akira Toda
for suggesting us the analysis of the general case, arbitrary $\gamma$ with two consumers, which will be the subject of our future research.} 
Condition \eqref{bmag} enables us to extend the result to HARA preferences and appraise the connection in ensuring uniqueness between
endowments, patience, and Bernoulli utility parameters affecting the concavity of the Bernoulli utility function.
\end{remark}


\subsection{The Toda-Walsh example with CRRA preferences}\label{CRRApart}
In this last part of the paper, we consider the example of multiple equilibria, provided by  Toda and Walsh \cite{towa},
of an economy with two goods and two consumers under the assumption of CRRA preferences and symmetric endowments.
Observe that this CRRA example can be obtained from HARA (see \eqref{uH}) by setting $b=0$
and 
\begin{equation}\label{beta}
\beta_1=(\frac{1-\alpha}{\alpha})^\gamma, \ \beta_2=(\frac{\alpha}{1-\alpha})^\gamma, \ 0 < \alpha < 1.
\end{equation} 
More precisely, preferences are given by

$$u_1(x_1,x_2)=\frac{1}{1-\gamma}(\alpha^\gamma x_1^{1-\gamma}+
(1-\alpha)^\gamma x_2^{1-\gamma})$$

$$u_2(x_1,x_2)=\frac{1}{1-\gamma}((1-\alpha)^\gamma x_1^{1-\gamma}+
\alpha^\gamma x_2^{1-\gamma})$$

and consumers' endowments are symmetric:
\begin{equation}\label{symm}
(e_1, f_1)=(e, 1-e) \ (e_2, f_2)=(1-e,e), \ 0<e<1.
\end{equation}



In this setting we obtain the following result.

\begin{theorem}\label{mainteor3}
In the CRRA symmetric case  with $\gamma =3$
the following facts hold true:
\begin{enumerate}
\item
 there is uniqueness of equilibria if and only if $(\alpha -3e)(2\alpha-1)> 0$;
 \item
 there is critical  equilibria if and only if $\alpha =3e$, $\alpha =\frac{1}{2}$;
 \item
 there are three equilibria if and only if $(\alpha -3e)(2\alpha-1)< 0$.
\end{enumerate}
\end{theorem}
\begin{proof}
In order to study the equilibria prices, we have to analyze the roots of \eqref{polic=2}
in this particular case.
Substituting \eqref{beta} and \eqref{symm} into \eqref{ABCD}, we get
$$P(q)=-\delta (\alpha, e)q^3+q^2-q+ \delta (\alpha, e),$$
where
$$\delta (\alpha, e):=\frac{\alpha^2-(2\alpha-1)e}{\alpha-\alpha^2}, \ 0<\alpha, e<1.$$
The  discriminant of this polynomial  (cfr. \eqref{dis}) is
$$\Delta =-(3\delta (\alpha, e)-1)^3(\delta (\alpha, e)+1).$$
Thus the result follows by noticing that $\Delta \geq 0 (<0)$ iff and only if $\delta (\alpha , e )\leq \frac{1}{3} (>\frac{1}{3})$, 
i.e. $(\alpha -3e)(2\alpha-1)\leq 0 (>0)$.
\end{proof}

Theorem \ref{mainteor3} complements \cite[Proposition 1]{towa}  by providing necessary and sufficient
conditions for uniqueness in the CRRA economy with relative risk aversion $\gamma=3$. 
It establishes a connection between the number of equilibria, their type (regular and critical),
endowments and utility weights.

We finally notice that for $\alpha =\frac{1}{7}$ and  $e=\frac{1}{49}$, one has $\delta (\alpha, e)=\frac{2}{7}$ and the polynomial $P(q)$
becomes
$-\frac{2}{7}q^3+q^2-q+\frac{2}{7}$ which has three  solutions, $\{\frac{1}{2},1,2\}$, which correspond, since $p=q^3$, 
to $\{\frac{1}{8}, 1, 8\}$ in accordance with \cite{towa}'s example.

\begin{remark}\label{remarkresearch}
Theorem \ref{mainteor3} has shown how this polynomial approach makes it easier to deal with the particular case $\gamma =3$.
We believe that this method deserves future investigation since it can be used to tackle the general case of $c$ consumers under CRRA and HARA 
preferences and arbitrary $\gamma$. 
\end{remark}

\end{document}